\documentclass{article}
\usepackage{etex}
\reserveinserts{28}
\usepackage{setspace, amsmath, amssymb, amsthm}
\usepackage[colorlinks = true]{hyperref}
\usepackage[margin = 2.5cm]{geometry}
\usepackage[english]{babel}
\usepackage[utf8x]{inputenc}
\usepackage{amsmath}
\usepackage{graphicx}
\usepackage[usenames,dvipsnames]{pstricks}
\usepackage{epsfig}
\usepackage{pst-grad} 
\usepackage{pst-plot} 
\usepackage{epstopdf}
\usepackage{subcaption}
\usepackage{tikz}
\usetikzlibrary{arrows}
\usetikzlibrary{calc}
\usetikzlibrary{fit}
\usepackage{pstricks-add}
\newtheorem{theorem}{Theorem}[section]
\newtheorem{lemma}[theorem]{Lemma}

\newtheorem{remark}[theorem]{Remark}
\newtheorem{corollary}[theorem]{Corollary}

\newtheorem{observation}{Observation}
\newtheorem{definition}{Definition}
\usepackage{tikz}

\newtheorem{example}{Example}

\title{Rank of weighted digraphs with blocks\thanks{This work is supported by the Joint NSFC-ISF Research Program (jointly funded by
the National Natural Science Foundation of China and the Israel Science Foundation (Nos.
11561141001, 2219/15).}}
\author{Ranveer Singh\thanks{ Technion-Israel Institute of Technology, Haifa 32000, Israel. Email: \texttt{singh@technion.ac.il} (Corresponding author)}\\
  Swarup Kumar Panda \thanks{Technion-Israel Institute of Technology, Haifa 32000, Israel. Email: \texttt{panda.iitg@gmail.com }}\\
Naomi Shaked-Monderer\thanks{The Max Stern Yezreel Valley College, Yezreel Valley 19300, Israel. Email: \texttt{nomi@technion.ac.il}}\\
Abraham Berman\thanks{Technion-Israel Institute of Technology, Haifa 32000, Israel. Email: \texttt{berman@technion.ac.il}}}

\begin{document}

        \maketitle

\begin{abstract}
Let $G$ be a digraph and $r(G)$ be its rank. Many interesting results on the rank of an undirected graph appear in the literature, but not much information about
the rank of a digraph is available. In this article, we study the rank of a digraph using the ranks of its blocks.  In particular, we define classes of digraphs, namely $r_2$-digraph, and $r_0$-digraph, for which the rank can be exactly determined in terms of the ranks of subdigraphs of the blocks. Furthermore, the rank of directed trees, simple biblock graphs, and some simple block graphs are studied.
\end{abstract}


\section{Introduction}
A well-known problem, proposed in 1957 by Collatz and Sinogowitz, is to characterize graphs with positive
nullity \cite{von1957spektren}. Nullity of graphs is applicable in various branches of science, in particular, quantum chemistry, H\"{u}ckel
molecular orbital theory \cite{gutman2011nullity}, \cite{lee1994chemical} and social network theory \cite{leskovec2010signed}. For more detail on the applications see \cite{gutman2011nullity}. Many significant results on the nullity of undirected graphs are available in the literature, see \cite{cvetkovic1980spectra, fiorini2005trees, gutman2001nullity, hu2008nullity, nath2007null, hu2008nullity, bapat2011note, berman2008upper, xuezhong2005nullity}.

Recently in \cite{gong2012nullity, gong2010nullity, fan2009nullity} nullity of undirected graphs was studied using cut-vertices and connected components. The results are used to calculate the nullity of line graphs of undirected graphs. Apart from the connected components, it turns out that the blocks are also interesting subgraphs of a graph, which can be utilized to know its determinant  \cite{singh2017characteristic, singh2017mathcal}. This motivates us to utilize blocks to determine the nullity or rank (for a square matrix $A$ of order $n$ nullity is equal to $n-r(A)$) of digraphs.

We construct some interesting classes of digraphs, in particular $r_2$-digraph and $r_0$-digraph, where the rank can be determined by the ranks of subdigraphs of the blocks. Furthermore, the ranks of directed trees, simple biblock graphs, and block graphs are determined.
\subsection{Notations and Preliminaries}
 A \emph{digraph} $G=(V(G), E(G))$ consists of a vertex set $V(G)$ and an arc set $E(G)\subseteq V(G) \times V(G)$. An arc $(u,u)$ is called a \emph{loop at the vertex $u$}.  A \emph{weighted digraph} is a digraph equipped with a weight function $f: E(G) \rightarrow \mathbb{C}$. 
    If $V(G) = \emptyset$ then, the digraph $G$ is called a \emph{null graph}. A \emph{subdigraph} of $G$ is a digraph $H$ such that $V(H) \subseteq V(G)$ and $E(H) \subseteq E(G)$. A subdigraph $H$ is an \emph{induced subdigraph} of $G$ if $u, v \in V(H)$ and $(u,v) \in E(G)$ implies $(u,v) \in E(H)$. 
    
A \emph{simple graph} $G=(V(G), E(G)),$ consists of a vertex set $V(G),$ and edge set $E(G),$ where each edge is an unordered pair $(u,v)$ of vertices, $u\neq v.$ With each directed graph $G$ we associate an underlying simple graph, with the same vertex set, and an edge (undirected) between two distinct vertices $u$ and $v$ if only if $(u,v)$ or $(v,u)$ is an arc of $G$. A \emph{cut-vertex} in a simple graph is a vertex whose removal increases the number of connected components. A \emph{block} in a simple graph is a maximal connected induced subgraph with no cut-vertex. A vertex in a digraph is a cut-vertex if it is a cut-vertex of its underlying simple graph. Similarly a subdigraph is a block of the digraph if it corresponds to a block in the underlying simple graph. A graph or digraph with no cut-vertex is called \emph{nonseparable}. A block of a digraph $G$ is \emph{pendant block} if it contains at most one cut-vertex of $G$.  In Figure \ref{bal} a digraph with seven blocks is presented. These blocks are the induced subdigraphs on the vertex-sets $\{v_1,v_2,v_3\}, \{v_1,v_4,v_5\}, \{v_1,v_6,v_7,v_8\}, \{v_6,v_{11},v_{12}\}, \{v_8,v_9,v_{10}\}, \{v_2,v_{14}\}, \{v_4,v_{13}\},$ respectively.   
  
For every weighted digraph $G$ there corresponds a matrix $A(G)$ with $a_{uv}=f(u,v)$ for every $(u,v)\in E(G)$ (in particular, $a_{uu}$ is the weight of the loop at $u,$ if it exists), and $a_{uv}=0$ otherwise.

Let $cs(M)$ denote the column space of a matrix $M$, that is, the space of all linear combinations of the columns of $M$. Similarly, $rs(M)$ denotes the row space of $M$.

Appending a row (column) to a matrix $M$ keeps its rank unchanged if the row (column) is in the row space (column space) of $M$, otherwise, the rank is increased by 1. Thus, increasing the size of a square matrix by one increases its rank by at most 2. It leads to the following observation.

\begin{observation} \label{ob1}
 Let $G$ be a digraph with a cut-vertex $v$. Let $H$ be a nonempty induced  subdigraph which includes $v$ such that there is no arc $(p,q)$ or $(q,p)$, where $p\in V(H\setminus v)$ and  $ q \in V(G\setminus H)$. Then one of the following three cases can occur.\\
\emph{CASE I:} $r(H)=r(H\setminus v)+2$.\\
\emph{CASE II:} $r(H)=r(H\setminus v)+0$.\\
\emph{CASE III:} $r(H)=r(H\setminus v)+1$.

\end{observation}
  Some typical examples of simple graphs and weighted digraphs for the above three cases are the following: 

\begin{example}

\emph{CASE I:} \begin{enumerate}
\item Simpe graph: Consider a nonsingular simple bipartite graph $H$ of order $n$. For example, a simple bipartite graph with a unique perfect matching is nonsingular. Note that $n$ has to be even as the eigenvalues of a bipartite graph are symmetric with respect to 0, thus if $n$ is odd then the bipartite graph is singular.  By Cauchy's interlacing eiegnvalues property,  if we remove any vertex $v$ from $H$ then the resulting bipartite graph will be singular having exactly one zero eigenvalue. Hence, $r(H\setminus v)=n-2.$ That is, $r(H)=r(H \setminus v)+2.$
\item Digraph: Consider a digraph $H$ on two vertices $u,v$ without loops. Let the arcs $(u,v), (v,u)$ have nonzero weights $\alpha_1, \alpha_2,$ respectively. Then $r(H)=r(H \setminus v)+2.$. 
\end{enumerate}

\emph{CASE II:} 
\begin{enumerate}
\item Simple graph: Consider a simple complete bipartite graph $K_{m,n}, m\ge1, n\ge1$. The nullity of $K_{m,n}$ is $m+n-2$. For $m>1,n>1$, if $H=K_{m,n},$ then $r(H)=2.$ If we remove a vertex $v$ from $H$, the rank will still be 2. That is, $r(H)=r(H\setminus v)+0=2.$

\item Digraph: Consider a digraph $H$ on two vertices $u,v$, with loop on $u$ but not on $v$. Let $H$ has only one arc $(u,v)$. Then $r(H)=r(H \setminus v)=1.$
\end{enumerate}

\emph{CASE III:} 
\begin{enumerate}
\item Simple graph: A simple complete graph $K_n, n\geq 2$ is nonsingular. Thus if $H=K_n, n\geq 3,$ then removal of any vertex $v$ results in a decrease of the rank by 1. That is, $r(H)=r(H\setminus v)+1.$

\item Digraph: Consider a digraph $H$ on two vertices $u,v$, without loops, and exactly one arc $(u,v)$. Then $r(H)=r(H \setminus v)+1.$
\end{enumerate}

\end{example} 

In Section \ref{secr2} we discuss CASE I and we define a new family of digraphs, $r_2$-digraphs, to obtain results on ranks of $r_2$-digraphs. In Section  \ref{secr0} we discuss CASE II, we define a new family of digraphs, $r_0$-digraphs, to obtain results on their ranks. In Section \ref{secr1} we discuss CASE III, and give some partial results, which includes results on ranks of block graphs discussed in \cite{singh2018nonsingular}.

\begin{theorem}\label{hy}
Consider the following square matrix
$$M=\begin{bmatrix}
\alpha & x^T\\ y & B
\end{bmatrix},$$ where $\alpha$ is a scalar, $x$ and $y$ are column vectors of order $n-1$, and $B$ is a square matrix of order $n-1$. Then \begin{enumerate}
\item $r(M)=r(B)+2$ if and only if  $x^T \notin rs(B), y\notin cs(B)$.
\item $r(M)=r(B)$ if and only if $\begin{bmatrix}
\alpha\\ y
\end{bmatrix}   \in cs \Big(\begin{bmatrix}
x^T \\ B
\end{bmatrix}\Big),  \begin{bmatrix}
\alpha & x^T
\end{bmatrix} \in rs\Big( \begin{bmatrix}
y & B
\end{bmatrix} \Big)$. 

\item $r(M)=r(B)+1$ if and only if one of the following hold \begin{enumerate}
\item  $x^T\in rs(B)$ and $ y \notin cs(B)$. 
\item $x^T \notin rs(B)$ and $ y \in cs(B).$ 
\item $x^T\in rs(B)$ and $y \in cs(B),$ but $ \begin{bmatrix}
\alpha & x^T
\end{bmatrix} \notin rs\Big(\begin{bmatrix}
y & B
\end{bmatrix} \Big)$ and $ \begin{bmatrix}
\alpha \\ y
\end{bmatrix} \notin cs \Big(\begin{bmatrix}
x^T \\ B
\end{bmatrix}\Big)$. 
\end{enumerate}
\end{enumerate} 
\end{theorem}
\begin{proof} 
   \begin{equation*}
 r(M)\leq r\Bigg( \begin{bmatrix}
x^T\\ B
\end{bmatrix}\Bigg)+1 \leq (r(B)+1)+1=r(B)+2.
\end{equation*}
\begin{enumerate} 
\item  Equality holds in the right inequality if and only if $x^T \notin rs(B)$ and in the left inequality if and only if  $ \begin{bmatrix}
\alpha\\ y
\end{bmatrix} \notin cs\Big(\begin{bmatrix}
x \\ B
\end{bmatrix}\Big)$. Similarly, $r(M)= r(B)+2$ if and only if $y \notin cs(B)$ and $\begin{bmatrix}
\alpha & x^T
\end{bmatrix} \notin rs \Big(\begin{bmatrix}
y & B
\end{bmatrix}\Big).$ Thus if $r(M)=r(B)+2,$ $x^T \notin rs(B)$ and $y \notin cs(B)$. Conversely, if $x^T \notin rs(B)$ and $y \notin cs(B),$ then  $\begin{bmatrix}
\alpha \\ y
\end{bmatrix} \notin cs\Big(\begin{bmatrix}
x^T \\ B
\end{bmatrix}\Big),$ which implies $r(M)=r(B)+2$. 

\item $r(M)=r\Bigg( \begin{bmatrix}
x^T\\ B
\end{bmatrix}\Bigg)$ if and only if $\begin{bmatrix}
\alpha \\ y
\end{bmatrix} \in cs\Big(\begin{bmatrix}
x^T \\ B
\end{bmatrix}\Big).$ And $r(B)=r\Bigg( \begin{bmatrix}
x^T\\ B
\end{bmatrix}\Bigg)$ if and only if $x^T \in rs(B)$. That is, $r(M)=r(B)$ if and only if $\begin{bmatrix}
\alpha \\ y
\end{bmatrix} \in cs \Big(\begin{bmatrix}
x^T \\ B
\end{bmatrix}\Big)$ and $x^T \in rs(B)$. Similarly, $r(M)=r(B)$ if and only if $\begin{bmatrix}
\alpha & x^T
\end{bmatrix} \in rs \Big(\begin{bmatrix}
y & B
\end{bmatrix}\Big)$ and $y \in cs(B)$. Hence, $r(M)=r(B),$ if and only if $\begin{bmatrix}
\alpha\\ y
\end{bmatrix} \in cs \Big(\begin{bmatrix}
x^T \\ B
\end{bmatrix}\Big)$,  $ \begin{bmatrix}
\alpha & x^T
\end{bmatrix} \in rs\Big(\begin{bmatrix}
y & B
\end{bmatrix}\Big).$
\item  If $x^T \in rs(B)$ then $r(B)=r\Bigg( \begin{bmatrix}
x^T\\ B
\end{bmatrix}\Bigg)$. If $y \notin cs(B)$ then $\begin{bmatrix}
\alpha \\ y
\end{bmatrix}\notin cs\Big(\begin{bmatrix}
x^T \\ B
\end{bmatrix}\Big)$. That is, if  $x^T\in rs(B)$ and $ y \notin cs(B)$ then $r(M)=r(B)+1$. 
Similarly, if $x^T \notin rs(B)$ and $ y \in cs(B).$  If $x^T \in rs(B)$ and $y \in cs(B)$, then $r(B)=r\Bigg( \begin{bmatrix}
x^T\\ B
\end{bmatrix}\Bigg)=r\Big( \begin{bmatrix}
y & B
\end{bmatrix}\Big)$. Which implies that  if  $ \begin{bmatrix}
\alpha & x^T
\end{bmatrix} \notin rs \Big(\begin{bmatrix}
y & B
\end{bmatrix}\Big)$, $ \begin{bmatrix}
\alpha \\ y
\end{bmatrix} \notin cs \Big( \begin{bmatrix}
x^T \\ B
\end{bmatrix}\Big)$, then $r(M)=r(B)+1.$

\end{enumerate} 
\end{proof}

 \begin{remark}\label{rmk1}
By part 1 of Lemma \ref{hy}, $r(M)=r(B)+2$ for some $\alpha \in \mathcal{R}$ if and only if $r(M)=r(B)+2$ for every $\alpha \in \mathcal{R}$.
\end{remark} 


\begin{figure}

     \begin{subfigure}[b]{0.35\textwidth}
             \begin{tikzpicture}[scale=0.6]

 \draw  node[draw,circle,scale=0.60] (1) at (0,0) {$v_1$};
 \draw  node[draw,circle,scale=0.60] (2) at (-1.5,1.5) {$v_2$};
 \draw  node[draw,circle,scale=0.60] (3) at (1.5,1.5) {$v_3$};
 \draw  node[draw,circle,scale=0.60] (4) at (-1.5,-1.5) {$v_4$};
\draw  node[draw,circle,scale=0.60] (5) at (1.5,-1.5) {$v_5$};
\draw  node[draw,circle,scale=0.60] (6) at (2.5,-1) {$v_6$};
\draw  node[draw,circle,scale=0.60] (8) at (2.5,1) {$v_8$};
\draw  node[draw,circle,scale=0.60] (7) at (5,0) {$v_7$};
\draw  node[draw,circle,scale=0.60] (9) at (2.5,2.5) {$v_9$};
\draw  node[draw,circle,scale=0.50] (10) at (4,2.5) {$v_{10}$};

\draw  node[draw,circle,scale=0.50] (11) at (2.5,-2.5) {$v_{11}$};
\draw  node[draw,circle,scale=0.50] (12) at (4,-2.5) {$v_{12}$};

\draw  node[draw,circle,scale=0.50] (13) at (-1.8,-2.8) {$v_{13}$};

\draw  node[draw,circle,scale=0.50] (14) at (-1.8,2.8) {$v_{14}$};

\tikzset{edge/.style = {-> = latex'}}
\draw[edge] (1) to (3);
\draw[edge] (3) to (2);
\draw[edge] (2) to (1);

\draw[edge] (10) to (8);
\draw[edge] (12) to (11);
\draw[edge] (12) to (6);
\draw[edge] (11) to (6);
\draw[edge] (13) to (4);
\draw[edge] (2) to (14);

\draw (4) to [out=200,in=170,looseness=8] (4);
\draw (11) to [out=330,in=300,looseness=8] (11);

\draw (7) to [out=330,in=300,looseness=8] (7);

\draw (10) to [out=380,in=350,looseness=8] (10);

\draw (3) to [out=150,in=120,looseness=8] (3);

\draw[edge] (8) to[bend left=8] (6);
\draw[edge] (6) to[bend left=8] (8);
\draw[edge] (6) to[bend left=8] (8);

\tikzset{edge/.style = {- = latex'}}
\draw[edge] (1) to (5);
\draw[edge] (5) to (4);
\draw[edge] (4) to (1);

\draw[edge] (1) to (8);
\draw[edge] (8) to (7);
\draw[edge] (7) to (6);
\draw[edge] (6) to (1);

\draw[edge] (8) to (9);
\draw[edge] (10) to (9);
\draw[edge] (1) to (7);

\end{tikzpicture}
     \caption{}
        \label{bal}
    \end{subfigure}\begin{subfigure}[b]{0.35\textwidth}
\begin{tikzpicture}[scale=0.6]

 \draw  node[draw,circle,scale=0.60] (1) at (0,0) {$v_1$};
 \draw  node[draw,circle,scale=0.60] (2) at (-1.5,1.5) {$v_2$};
 \draw  node[draw,circle,scale=0.60] (3) at (1.5,1.5) {$v_3$};
 \draw  node[draw,circle,scale=0.60] (4) at (-1.5,-1.5) {$v_4$};
\draw  node[draw,circle,scale=0.60] (5) at (1.5,-1.5) {$v_5$};
\draw  node[draw,circle,scale=0.60] (6) at (2.5,-1) {$v_6$};
\draw  node[draw,circle,scale=0.60] (8) at (2.5,1) {$v_8$};
\draw  node[draw,circle,scale=0.60] (7) at (5,0) {$v_7$};
\draw  node[draw,circle,scale=0.60] (9) at (2.5,2.5) {$v_9$};
\draw  node[draw,circle,scale=0.50] (10) at (4,2.5) {$v_{10}$};

\draw  node[draw,circle,scale=0.50] (11) at (2.5,-2.5) {$v_{11}$};
\draw  node[draw,circle,scale=0.50] (12) at (4,-2.5) {$v_{12}$};

\draw  node[draw,circle,scale=0.50] (13) at (-1.8,-2.8) {$v_{13}$};

\draw  node[draw,circle,scale=0.50] (14) at (-1.8,2.8) {$v_{14}$};

\draw  node[draw,circle,scale=0.50] (15) at (-0.5,2.8) {$v_{15}$};
\draw  node[draw,circle,scale=0.50] (16) at (-1.8,0) {$v_{16}$};
\draw  node[draw,circle,scale=0.50] (17) at (-0.5,-2.8) {$v_{17}$};
\draw  node[draw,circle,scale=0.50] (18) at (4.5,1.5) {$v_{18}$};
\draw  node[draw,circle,scale=0.50] (19) at (4.5,-1.5) {$v_{19}$};

\tikzset{edge/.style = {-> = latex'}}
\draw[edge] (1) to (3);
\draw[edge] (3) to (2);
\draw[edge] (2) to (1);

\draw[edge] (10) to (8);
\draw[edge] (12) to (11);
\draw[edge] (12) to (6);
\draw[edge] (11) to (6);
\draw[edge] (13) to (4);
\draw[edge] (2) to (14);

\draw (4) to [out=200,in=170,looseness=8] (4);
\draw (11) to [out=330,in=300,looseness=8] (11);

\draw (7) to [out=330,in=300,looseness=8] (7);

\draw (10) to [out=380,in=350,looseness=8] (10);

\draw (2) to [out=200,in=170,looseness=8] (2);

\draw (3) to [out=150,in=120,looseness=8] (3);

\draw[edge] (8) to[bend left=8] (6);
\draw[edge] (6) to[bend left=8] (8);
\draw[edge] (6) to[bend left=8] (8);
\draw[edge] (2) to[bend left=8] (15);
\draw[edge] (15) to[bend left=8] (2);
\draw[edge] (1) to[bend left=8] (16);
\draw[edge] (16) to[bend left=8] (1);
\draw[edge] (6) to[bend left=8] (19);
\draw[edge] (19) to[bend left=8] (6);

\tikzset{edge/.style = {- = latex'}}
\draw[edge] (1) to (5);
\draw[edge] (5) to (4);
\draw[edge] (4) to (1);

\draw[edge] (1) to (8);
\draw[edge] (8) to (7);
\draw[edge] (7) to (6);
\draw[edge] (6) to (1);

\draw[edge] (8) to (9);
\draw[edge] (10) to (9);
\draw[edge] (1) to (7);
\draw[edge] (18) to (8);
\draw[edge] (17) to (4);

\end{tikzpicture}
     \caption{}\label{r2di}
    \end{subfigure}          \begin{subfigure}[b]{0.22\textwidth}
\begin{tikzpicture}[scale=0.6]

 \draw  node[draw,circle,scale=0.60] (1) at (0,0) {$v_1$};
 \draw  node[draw,circle,scale=0.60] (2) at (-1.5,1.5) {$v_2$};

 \draw  node[draw,circle,scale=0.60] (3) at (-1.5,-1.5) {$v_3$};
\draw  node[draw,circle,scale=0.60] (4) at (1.5,-1.5) {$v_4$};
\draw  node[draw,circle,scale=0.60] (5) at (2.5,0) {$v_5$};

\draw  node[draw,circle,scale=0.60] (6) at (5,0) {$v_6$};
\draw  node[draw,circle,scale=0.60] (7) at (2.5,2.5) {$v_7$};
\draw  node[draw,circle,scale=0.60] (8) at (4,2.5) {$v_{8}$};

\draw  node[draw,circle,scale=0.60] (9) at (0.6,3) {$v_{9}$};

\draw  node[draw,circle,scale=0.50] (10) at (4,-1.8) {$v_{10}$};

\tikzset{edge/.style = {-> = latex'}}

\draw[edge] (1) to[bend left=8] (2);
\draw[edge] (2) to[bend left=8] (1);

\draw[edge] (3) to[bend left=8] (1);

\draw[edge] (1) to[bend left=8] (7);
\draw[edge] (7) to[bend left=8] (1);

\draw[edge] (7) to[bend left=8] (5);
\draw[edge] (5) to[bend left=8] (7);

\draw[edge] (5) to[bend left=8] (8);
\draw[edge] (8) to[bend left=8] (5);

\draw[edge] (8) to[bend left=8] (6);
\draw[edge] (6) to[bend left=8] (8);

\draw[edge] (4) to[bend left=8] (5);
\draw[edge] (5) to[bend left=8] (4);

\draw[edge] (7) to[bend left=8] (9);
\draw[edge] (9) to[bend left=8] (7);
\draw[edge] (5) to[bend left=8] (10);

\draw (3) to [out=330,in=300,looseness=8] (3);
\draw (1) to [out=330,in=300,looseness=8] (1);
\draw (8) to [out=380,in=350,looseness=8] (8);
\draw (5) to [out=380,in=350,looseness=8] (5);

\end{tikzpicture}
     \caption{}\label{rti}
    \end{subfigure}

\caption{The digraph in (a) is extented to an $r_2$-digraph in (b) by adding $\tilde{nc}$-edges at the cut-vertices. (c) An $r_2$-tree digraph. In this figure the undirected edges are simple edges, which are equivalent to two opposite arcs of same weight.} \label{r2fig}
\end{figure}
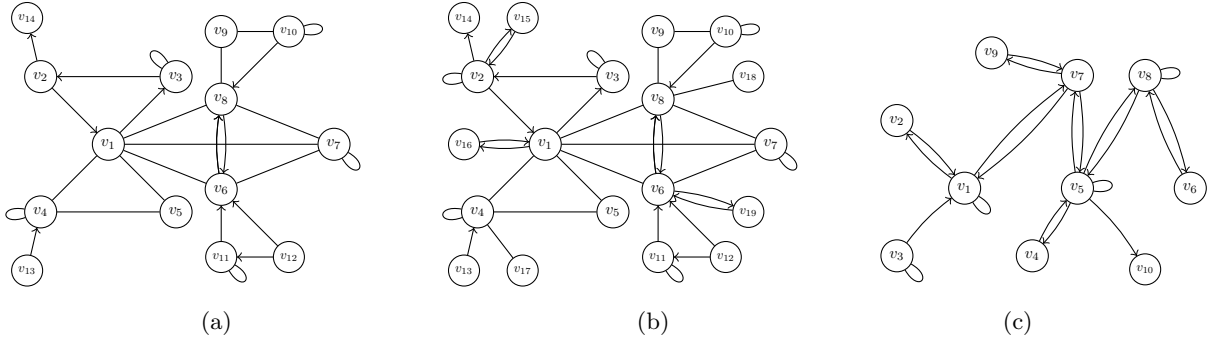

%
%
%
%
%
%

\section{CASE I} \label{secr2}

 \begin{theorem} \label{+2}
 Let $G$ be a digraph with a cut-vertex $v$. Let $H$ be a nonempty induced  subdigraph which includes $v$ such that there is no arc $(p,q)$ or $(q,p)$, where $p\in V(H\setminus v)$ and  $ q \in V(G\setminus H)$. If $r(H)=r(H\setminus v)+2$, then
    $$r(G)=r(H\setminus v)+ r(G\setminus H)+2.$$
    \end{theorem}
    \begin{proof} With suitable reordering of the vertices in $G$, we can write,
    $$A(G)=\begin{bmatrix} \alpha & x^T & w^T \\ y & A(H \setminus v) & O \\ z & O^T & A(G \setminus H)\end{bmatrix},$$ where the first row and the first column correspond to the cut-vertex $v$, $\alpha$ is the weight of the loop at $v$, and $O$ denotes the zero matrix of suitable order. Clearly the rank of the matrix $$\begin{bmatrix} A(H \setminus v) & O \\  O^T & A(G \setminus H)\end{bmatrix}$$ is $r(H\setminus v)+r(G\setminus H)$. Consider the submatrix $$M=\begin{bmatrix}  x^T & w^T \\ A(H \setminus v) & O \\  O^T & A(G \setminus H)\end{bmatrix}.$$ The rank of $M$ can be at most $r(H\setminus v)+r(G\setminus H)+1$. By the hypothesis and using Theorem \ref{hy}(1), $x^T \notin rs(A(H\setminus v))$. Thus the first row of $M$ is not in the row space of the rest of the matrix. Hence the rank of $M$ is $r(H\setminus v)+r(G\setminus H)+1$. Now, consider $A(G)$. Its rank can be at most $r(M)+1$. Again using Theorem \ref{hy}(1)  and by the hypothesis, the column vector $y \notin cs(A(H\setminus v))$, which implies that the first column of $A(G)$ is not  in the column space of $M$. Hence $$r(A(G))=r(M)+1=r(H\setminus v)+r(G\setminus H)+2.$$ Which proves the result.
    \end{proof}

\subsection{Tree digraph}
 A directed edge or arc from a vertex $v_1$ to a vertex $v_2$ of weight $\alpha$ is denoted by $\overrightarrow{v_1\alpha v_2}$. A \emph{simple weighted edge} of weight $\alpha$ between $v_1$ and $v_2$ is an edge, where $\overrightarrow{v_1\alpha v_2}, \overrightarrow{v_2\alpha v_1}$, and no loops on $v_1$ and $v_2$. In general an edge between the vertices $v_1, v_2$ in a weighted directed graph $G$ is a set of two arcs, $\overrightarrow{v_1\alpha_1 v_2}, \overrightarrow{v_2\alpha_2 v_1}$. We now define a few more types of edges for the purpose of our study. Let $v$ be a cut-vertex and $u$ be a noncut-vertex of the digraph $G$. An \emph{$\tilde{nc}$-edge} between $v$ and $u$ is an edge, where $\overrightarrow{v\alpha_1 u}, \overrightarrow{u \alpha_2 v}$ ($\alpha_1, \alpha_2$ are nonzero), and there is no loop on the noncut-vertex $u$. Similarly, an \emph{$\tilde{nc}$-arc} between $v$ and $u$ is an arc either from $v$ to $u$ or $u$ to $v$ and no loop on $u$. Analogously, an \emph{$nc$-edge} between $v$ and $u$ is an edge, where $\overrightarrow{v\alpha_1 u}, \overrightarrow{u \alpha_2 v},$ ($\alpha_1, \alpha_2$ are nonzero), and there is a loop on $u$. Finally, an \emph{$nc$-arc} between $v$ and $u$ is an arc either from $v$ to $u$ or $u$ to $v$ and a loop on $u$. \emph{Matching} in a digraph is a set of vertex disjoint pairs of arcs on the same vertices.

For an undirected tree having maximum matching of size $q$, the rank is equal to $2q$ \cite{cvetkovic1972algebraic}. We show that this result is also true for some categories of tree digraphs.  Let $T$ denote a tree digraph obtained from an undirected tree by replacing its each simple edge $(u,v)$ by arcs $\overrightarrow{u\alpha_1 v}, \overrightarrow{v\alpha_2 u}$ of arbitrary nonzero weights $\alpha_1, \alpha_2.$ We will call such a tree digraph $T$ a loopless bi-arc tree. If loops are allowed only at cut-vertices then we call $T$ as a cut-loop bi-arc tree. That is, in a cut-loop bi-arc tree, any noncut vertex does not have a loop while any cut-vertex may or may not have a loop. We consider a class of tree digraphs having loops and single arcs on the vertices in the next subsection.

\begin{theorem} \label{tt}
Let $T$ be a loopless bi-arc tree having the maximum matching of size $q$. Then $r(T)=2q$.
\end{theorem}
\begin{proof}
We denote by $T_n$ any bi-arc tree with $n$ vertices. We will prove the result by using induction on $n$. For $n=1$, clearly the result is true. Assume that the result is true for every $T_n$, $n\geq 1$. Consider a tree $T_{n+1}$. Let $e_{uv}=(u,v)$ be a pendant edge of $T_{n+1},$ $v$ the cut-vertex.  Let $C_1, \ldots, C_k$ be the components (bi-arc trees of order less than $n$) of $T_{n+1}\setminus (u,v)$ having maximum matching of size
$m_1,\ldots,m_k$, respectively. Adding to the maximal matchings of $C_1,\hdots, C_k,$ the edge $(u,v)$ yields a matching in $T_{n+1}$ of size $\sum_{i=1}^{k} m_i+1.$ This is a maximal matching in $T_{n+1},$ since any matching in $T_{n+1}$ consists of at most one edge incident with $v$ and matchings in $C_1,\hdots, C_k,$ and thus has size of at most $\sum_{i=1}^{k} m_i+1.$ By the induction hypothesis $r(T_{n+1}\setminus (u,v))=2\sum_{i=1}^km_i$. Using Theorem \ref{+2}, $r(T_{n+1})=r(e_{uv}\setminus v)+r(T_{n+1}\setminus (u,v))+2=0+2(\sum_{i=1}^km_i)+2=2(\sum_{i=1}^km_i+1).$  It proves the result. \end{proof} 

%
%
%
%

\begin{lemma} \label{2rin}
Consider $m$ vertices $v_1,\hdots, v_m$ of a digraph $G$. Then $r(G)=r(G\setminus\{v_1,\hdots,v_m\})+2m$ if and only if for every subset $S$ of $\{v_1,\hdots, v_m\}$, $r(G)=r(G\setminus S)+2|S|.$ 
\end{lemma}
\begin{proof}
Without loss of generality let us consider a subset $S=\{v_1,\hdots v_r\},$ $r\leq m$ of $\{v_1,\hdots v_m\}$. Assume that $r(G)\neq r(G\setminus S)+2r$ but $r(G)=r(G\setminus\{v_1,\hdots,v_m\})+2m$. Note that as deletion of vertex from $G$ results in deletion of a row and a column in corresponding matrix hence the rank can be decrease at most by 2. Thus $r(G)\neq r(G\setminus S)+2r$ implies $r(G)< r(G\setminus S)+2r$. If we increase the size of $S$ by adding vertices $v_{r+1}, v_{r+2},\hdots v_m$, that is $S= \{v_1,\hdots v_m\}$. Then further there can be at most $2(m-r)$ reduction in the rank. That is now $r(G)< r(G\setminus S)+2r+2(m-r)=r(G\setminus S)+2m.$ Which is a contradiction, hence the result follows. 
\end{proof}

Let $G$ be a digraph with $m$ cut-vertices and $k$ blocks $B_1, B_2,\hdots, B_k$ having $m_1,m_2,\hdots, m_k$ cut-vertices of $G,$ respectively. 
Then $\breve{B_i}$ denotes the induced subgraph of $B_i$ on the noncut-vertices, $i=1,2,\hdots,k$. And $G^q$ denotes a digraph obtained after removing of $q(\leq m)$ cut-vertices from $G$.

\begin{definition}
A block induced pendant subdigraph of $G^q$ is a maximal subdigraph of any block of $G$ such that it has exactly one cut vertex of $G$.
\end{definition}

\begin{theorem}\label{mdt}
If $r(B_i)=r(\breve{B_i})+2m_i, V(\breve{B_i})\neq \phi, i=1,2,\hdots k,$ then $r(G)=\sum_{i=1}^k r(\breve{B_i})+2m.$
\end{theorem}
\begin{proof}
Let $G^1,\ldots,G^m$ be  subdigraph of $G$ such that $G^{i+1}$ is obtained from $G^{i}$ by deleting exactly one cut vertex of
$G$ for $i=1,\ldots,m$. In $G,$ for $i=1,\hdots,k,$ $r(\breve{B_i})=r(B_i)-2m_i$. Using Lemma \ref{2rin}, for $i=1,\hdots,k$ removal of $r (\leq m_i)$ cut-vertices from $B_i$ decrease its rank by $2r$. Hence in $G^i$ any pendant subdigraph $B$ with the cut-vertex $v$ of $G$ satisfies $r(B)=r(B\setminus v)+2.$  Applying Theorem \ref{+2} for a block induced pendant subdigraph of each $G^{i}$, the result follows.
\end{proof}

\begin{definition} \textbf{$r_2$-block:} An $r_2$-block of a digraph is a pendant block $B$, with the cut-vertex $v,$ such that, $r(B)=r(B\setminus v)+2.$
\end{definition}

\begin{definition} \textbf{$r_2$-digraph:} A digraph is called $r_2$-digraph if at each cut-vertex of the digraph there is an $r_2$-block.
\end{definition}

Notice that any digraph can be extended to an $r_2$-digraph by adding $r_2$-blocks at all the cut-vertices. Moreover, an $r_2$-digraph can be extended to a higher order $r_2$-digraph by coalescing $r_2$-blocks at arbitrary vertices. Thus a nonseparable digraph can also be converted to an $r_2$-digraph. The smallest $r_2$-block in a digraph could be an $\tilde{nc}$-edge. A digraph in Figure \ref{bal} is extended to an $r_2$-digraph in Figure \ref{r2di} by attaching $\tilde{nc}$-edges at the cut-vertices.

\begin{theorem} \label{pen}
Let $G$ be an $r_2$-digraph having $k$ blocks, $B_1, B_2,\hdots, B_k,$ and $m$ cut-vertices. Then  $$r(G)=\sum_{i=1}^k r(\breve{B_i})+2m.$$
\end{theorem}
\begin{proof}
The proof follows by using Theorem \ref{+2} for each $r_2$-block of $G$.
\end{proof}

\begin{corollary} Adding loops or weights on the loops at the cut-vertices does not change the rank of an $r_2$-digraph. \end{corollary}
\begin{proof}
By Remark \ref{rmk1}, the weights of the loops at the cut-vertices did not play any role in Theorem \ref{+2} and Theorem \ref{pen}, and hence the result follows.
\end{proof}

\begin{theorem}\label{genr2}
Let $G$ be an $r_2$-digraph with $m$ cut-vertices. Consider $s$ digraphs $W_1, W_2,\hdots, W_s$. Let an $\tilde{nc}$-edge be added between one arbitrary vertex of $W_i$ and one arbitrary cut-vertex of $G$ for $i=1,\hdots s$. Let $G'$ be the resulting digraph. Then $$r(G')=\sum_{i=1}^{k}r(\breve{B})+\sum_{i=1}^{s}r(W_i)+2m.$$
\end{theorem}
\begin{proof}
The proof follows by Theorem \ref{+2} for each $r_2$-block of $G$.
\end{proof}

\begin{corollary} \label{cr2}
Let $G$ be an $r_2$-digraph. On adding edges or arcs to its cut-vertices, the rank increases by the number of $nc$-edges or $nc$-arcs. Thus the rank is unchanged by the addition of simple edges at its cut-vertices.  \end{corollary}
\begin{proof}
The proof directly follows by Theorem \ref{genr2}.
\end{proof}

We will now see some more categories of tree digraphs.
\begin{definition}\textbf{$r_2$-Tree Digraph:} Let $T$ be a cutloop bi-arc tree. If at each cut-vertex of $T$ there exists an $\tilde{nc}$-edge then it is called an $r_2$-tree digraph.
\end{definition}

\begin{corollary}
Let $T$ denote an $r_2$-tree digraph with $s$ $nc$-edges or $nc$-arcs. If $q$ is the maximum matching in $T$, then $r(T)=2q+s.$
\end{corollary}
\begin{proof}
It follows by Theorem \ref{tt} and Corollary \ref{cr2}.
\end{proof}

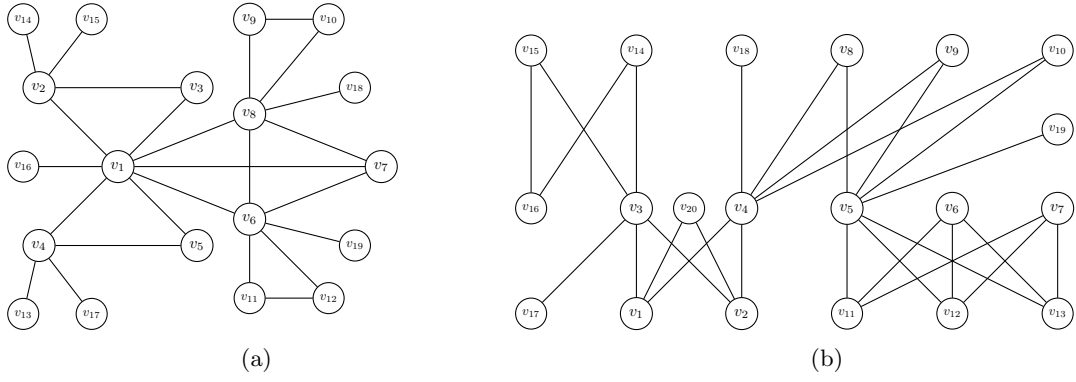
\begin{figure}

     \begin{subfigure}[b]{0.4\textwidth}
\begin{tikzpicture} [scale=0.7] [->,>=stealth',shorten >=1pt,auto,node distance=4cm,
                thick,main node/.style={circle,draw,font=\Large\bfseries}]

 \draw  node[draw,circle,scale=0.60] (1) at (0,0) {$v_1$};
 \draw  node[draw,circle,scale=0.60] (2) at (-1.5,1.5) {$v_2$};
 \draw  node[draw,circle,scale=0.60] (3) at (1.5,1.5) {$v_3$};
 \draw  node[draw,circle,scale=0.60] (4) at (-1.5,-1.5) {$v_4$};
\draw  node[draw,circle,scale=0.60] (5) at (1.5,-1.5) {$v_5$};
\draw  node[draw,circle,scale=0.60] (6) at (2.5,-1) {$v_6$};
\draw  node[draw,circle,scale=0.60] (8) at (2.5,1) {$v_8$};
\draw  node[draw,circle,scale=0.60] (7) at (5,0) {$v_7$};
\draw  node[draw,circle,scale=0.60] (9) at (2.5,2.8) {$v_9$};
\draw  node[draw,circle,scale=0.50] (10) at (4,2.8) {$v_{10}$};

\draw  node[draw,circle,scale=0.50] (11) at (2.5,-2.5) {$v_{11}$};
\draw  node[draw,circle,scale=0.50] (12) at (4,-2.5) {$v_{12}$};

\draw  node[draw,circle,scale=0.50] (13) at (-1.8,-2.8) {$v_{13}$};

\draw  node[draw,circle,scale=0.50] (14) at (-1.8,2.8) {$v_{14}$};

\draw  node[draw,circle,scale=0.50] (15) at (-0.5,2.8) {$v_{15}$};
\draw  node[draw,circle,scale=0.50] (16) at (-1.8,0) {$v_{16}$};
\draw  node[draw,circle,scale=0.50] (17) at (-0.5,-2.8) {$v_{17}$};
\draw  node[draw,circle,scale=0.50] (18) at (4.5,1.5) {$v_{18}$};
\draw  node[draw,circle,scale=0.50] (19) at (4.5,-1.5) {$v_{19}$};

\tikzset{edge/.style = {- = latex'}}
\draw[edge] (1) to (3);
\draw[edge] (3) to (2);
\draw[edge] (2) to (1);

\draw[edge] (2) to (15);
\draw[edge] (16) to (1);
\draw[edge] (8) to (6);
\draw[edge] (6) to (19);
\draw[edge] (10) to (8);
\draw[edge] (12) to (11);
\draw[edge] (12) to (6);
\draw[edge] (11) to (6);
\draw[edge] (13) to (4);
\draw[edge] (2) to (14);

\draw[edge] (1) to (5);
\draw[edge] (5) to (4);
\draw[edge] (4) to (1);

\draw[edge] (1) to (8);
\draw[edge] (8) to (7);
\draw[edge] (7) to (6);
\draw[edge] (6) to (1);

\draw[edge] (8) to (9);
\draw[edge] (10) to (9);
\draw[edge] (1) to (7);
\draw[edge] (18) to (8);
\draw[edge] (17) to (4);

\end{tikzpicture}
     \caption{}\label{r2block}
    \end{subfigure}    \begin{subfigure}[b]{0.5\textwidth}
 \begin{tikzpicture}[scale=0.7] [->,>=stealth',shorten >=1pt,auto,node distance=4cm,
                thick,main node/.style={circle,draw,font=\Large\bfseries}]
                       \tikzset{vertex/.style = {shape=circle,draw}}
 \tikzset{edge/.style = {- = latex'}}
\draw  node[draw,circle,scale=0.60] (1) at (0,0) {$v_1$};
\draw  node[draw,circle,scale=0.60] (2) at (2,0) {$v_2$};
\draw  node[draw,circle,scale=0.60] (3) at (0,2) {$v_3$};
\draw  node[draw,circle,scale=0.60] (4) at (2,2) {$v_4$};

\draw  node[draw,circle,scale=0.50] (20) at (1,2) {$v_{20}$};

\draw  node[draw,circle,scale=0.60] (5) at (4,2) {$v_5$};
\draw  node[draw,circle,scale=0.60] (6) at (6,2) {$v_6$};
\draw  node[draw,circle,scale=0.60] (7) at (8,2) {$v_7$};
\draw  node[draw,circle,scale=0.60] (8) at (4,5) {$v_8$};
\draw  node[draw,circle,scale=0.60] (9) at (6,5) {$v_9$};
\draw  node[draw,circle,scale=0.50] (10) at (8,5) {$v_{10}$};

\draw  node[draw,circle,scale=0.50] (11) at (4,0) {$v_{11}$};
\draw  node[draw,circle,scale=0.50] (12) at (6,0) {$v_{12}$};
\draw  node[draw,circle,scale=0.50] (13) at (8,0) {$v_{13}$};

\draw  node[draw,circle,scale=0.50] (14) at (0,5) {$v_{14}$};

\draw  node[draw,circle,scale=0.50] (15) at (-2,5) {$v_{15}$};

\draw  node[draw,circle,scale=0.50] (16) at (-2,2) {$v_{16}$};

\draw  node[draw,circle,scale=0.50] (17) at (-2,0) {$v_{17}$};
\draw  node[draw,circle,scale=0.50] (18) at (2,5) {$v_{18}$};
\draw  node[draw,circle,scale=0.50] (19) at (8,3.5) {$v_{19}$};

 \draw[edge] (1) to (3);
 \draw[edge] (1) to (4);
 \draw[edge] (2) to (3);
 \draw[edge] (2) to (4);
 \draw[edge] (4) to (8);
 \draw[edge] (4) to (9);
 \draw[edge] (4) to (10);
 \draw[edge] (5) to (8);
 \draw[edge] (5) to (9);
 \draw[edge] (5) to (10);
 \draw[edge] (1) to (20);
 \draw[edge] (2) to (20);

 \draw[edge] (5) to (11);
 \draw[edge] (5) to (12);
 \draw[edge] (5) to (13);
 \draw[edge] (6) to (11);
 \draw[edge] (6) to (12);
 \draw[edge] (6) to (13);
 \draw[edge] (7) to (11);
 \draw[edge] (7) to (12);
 \draw[edge] (7) to (13);
 \draw[edge] (3) to (14);
 \draw[edge] (3) to (15);
 \draw[edge] (16) to (15);
 \draw[edge] (14) to (16);

 \draw[edge] (3) to (17);
 \draw[edge] (4) to (18);
 \draw[edge] (5) to (19);
 \end{tikzpicture}
  \caption{}\label{r2biblock}
    \end{subfigure}

\caption{(a) $r_2$-block graph. (b) $r_2$-biblock graph.} \label{r2fig}

\end{figure}

  A simple graph is called a \emph{block graph} if each of its blocks is a simple complete graph. When each block of a simple graph is simple complete bipartite graph then it is called a \emph{biblock graph}.  A block graph is called an \emph{$r_2$-block graph} if there exists a simple pendant edge at each of its cut-vertices. An example of an $r_2$-block graph is given in Figure \ref{r2block}. Similarly, A biblock graph is called an \emph{$r_2$-biblock graph} if there exists a simple pendant edge at each of its cut-vertices. An example of an $r_2$-biblock graph is given in Figure \ref{r2biblock}.

\begin{corollary}\cite{singh2018nonsingular}(Theorem 3.5)
Let $G$ be an $r_2$-block graph with $m$ cut-vertices and $k$ blocks. Choose $m$ pendant edges, one at each cut-vertex. If each of the remaining $n-k$ blocks has at least two noncut-vertices, then $G$ is nonsingular.  
\end{corollary}

\begin{corollary}
Let $G$ be an $r_2$-biblock graph with $m$ cut-vertices and $k$ blocks. Choose $m$ pendant edges, one at each cut-vertex. If each of the remaining $n-k$ blocks has at least two noncut-vertices in different partition, then $r(G)=2k.$
\end{corollary}
\begin{proof} Let $B_1,B_2,\hdots, B_m$ be the blocks which are $m$ pendant edges selected one from each cut-vertex of $G$. Let $B_{m+1}, B_{m+2},\hdots, B_k$ be the rest of the blocks in $G$. Using Theorem \ref{pen} $$r(G)=\sum_{i=m+1}^k r(\breve{B_i})+2m.$$
Note that any $\breve{B_i}$ is complete bipartite graph having at least 2 vertices, thus it has rank 2. Hence, $$r(G)=2(k-m)+2m=2k.$$ This completes the proof.
\end{proof}

\section{CASE II} \label{secr0}

Let $G_1$ be any subdigraph of a digraph $G$. Let $v$ be a cut-vertex of $G$ with a loop of weight $\alpha$. By $v[G_1]_{in}$, we denote the column vector consisting of the weights of all the incoming edges from $G_1$ to $v$. Similarly, by $v[G_1]_{out}$, we denote the row vector consisting of weights of all the outgoing edges from $v$ to $G_1$.
\begin{theorem} \label{r_0}
    Let $G$ be a digraph with a cut-vertex $v$. Let $H$ be a nonempty induced  subdigraph which includes $v$ such that there is no arc $(p,q)$ or $(q,p)$, where $p\in V(H\setminus v)$ and  $ q \in V(G\setminus H)$. If $r(H)=r(H\setminus v)$, and any of the following happens,
\begin{enumerate}
\item $\alpha=0.$
\item $v[G\setminus H]_{in}$ or $v[G\setminus H]_{out}$ is independent of $A(G\setminus H)$.
\end{enumerate}
   Then
    $$r(G)=r(H\setminus v)+ r(G\setminus (H\setminus v)).$$
\end{theorem}
\begin{proof}
With suitable relabelling of the vertices, we can write
$$A(G)=\begin{bmatrix}
A(H\setminus v) & x & O\\ y^T & \alpha & w^T \\ O^T & z & A(G \setminus H)
\end{bmatrix}.$$ Notice that by the hypothesis the row vector $[y^T \ \alpha] \in rs([A(H\setminus v) \ x])$. Similarly, the column vector $\begin{bmatrix}
x \\ \alpha \end{bmatrix}\in cs \Bigg(\begin{bmatrix}
A(H\setminus v) \\ y^T
\end{bmatrix}  \Bigg)$.

Using elementary  operations, $A(G)$ can be transformed to the matrix $$\begin{bmatrix}
A(H\setminus v) & o & O\\ o^T & 0 & w^T \\ O^T & z & A(G \setminus H)
\end{bmatrix}.$$ Clearly the result follows for $\alpha=0$.
When $w^T \notin rs(A(G\setminus H))$ or $z \notin cs(A(G\setminus H)),$ then for any value of $\alpha$,
$$r\Bigg(\begin{bmatrix}
 \alpha &  w^T \\ z & A(G \setminus H)
\end{bmatrix}\Bigg)=r\Bigg( \begin{bmatrix}
 0 &  w^T \\ z & A(G \setminus H)
\end{bmatrix} \Bigg).$$ Hence the result follows.
\end{proof}

 \begin{definition} Let $B$ be a block in a digraph. If for every cut-vertex $v$ in $B,$  $r(B)=r(B\setminus v),$ then $B$ is called an $r_0$ block. \end{definition}

 \begin{definition} Let $G$ be a digraph having $k$ blocks. If at least $k-1$ blocks are $r_0$ blocks, then $G$ is called $r_0$-digraphs. \end{definition}

\begin{theorem}\label{r0f}
Let $G$ be a graph having no loops at the cut-vertices. If $G$ is $r_0$-digraph, then $r(G)=\sum_{i=1}^{k} r(B_i)$.
\end{theorem}

\begin{proof}
 We prove the result by induction on the number $k$ of blocks. For $k=1$ the result is trivial. Suppose the result holds for $k$. Consider a digraph having $k+1$ blocks. Since $G$ has at least two pendant blocks, there exists a pendant block, without loss of generality $B_1$ which is an $r_0$-block with the cut-vertex $v$. Then using Theorem \ref{r_0} $$r(G)=r(B_1\setminus v)+r(G\setminus (B_1\setminus v))=r(B_1)+r(G\setminus (B_1\setminus v)).$$ Since, $G\setminus (B_1\setminus v)$ has $k$ blocks, all of them, except for at most one, $r_0$-blocks, by the induction hypothesis $$r\bigg(G\setminus (B_1\setminus v)\bigg)=\sum_{i=2}^{k+1} r(B_i).$$ Hence $r(G)=\sum_{i=1}^{k+1}r(B_i),$ which proves the result.
\end{proof}

\begin{corollary}
Let $G$ be a biblock graph having $k$ blocks. If each block has at least two noncut-vertices in different partition sets, then $r(G)=2k.$
\end{corollary}
\begin{proof}
Note that such a biblock graph $G$ is an $r_0$-digraph. Hence the result follows from the Theorem \ref{r0f}.
\end{proof}

\section{CASE III}\label{secr1}
Let $H$ be any subdigraph of a digraph $G$ having a cut-vertex $v$ with a loop of weight $\alpha$. By $v[\alpha \ H]_{in}$ we denote the column vector consisting of the weight of the loop at $v$ and weights of all the incoming edges from $H$ to $v$. Similarly, by $v[\alpha \ H]_{out}$ we denote the column vector consisting of the weight of the loop at $v$ and the weights of all the outgoing edges from $v$ to $H$.

\begin{theorem}
  Let $G$ be a digraph with a cut-vertex $v$. Let $H$ be a nonempty induced  subdigraph which includes $v$ such that there is no arc $(p,q)$ or $(q,p)$, where $p\in V(H\setminus v)$ and  $ q \in V(G\setminus H)$. If $r(H)=r(H\setminus v)+1$, then \begin{enumerate}
\item If $v[H\setminus v]_{in} \in cs(A(H\setminus v))$,
$$ r(G)=\begin{cases}
r(H\setminus v)+ r(G\setminus H)+1 & \text{if $v[G\setminus H]_{in}\in cs(A(G\setminus H)),$ }\\ r(H\setminus v)+ r(G\setminus H)+2 & \text{if $v[G\setminus H]_{in} \notin cs(A(G\setminus H)).$ }
\end{cases}$$

\item If $v[H\setminus v]_{out} \in rs(A(H\setminus v))$, $$ r(G)=\begin{cases}
r(H\setminus v)+ r(G\setminus H)+1 & \text{if $v[G\setminus H]_{out}\in rs(A(G\setminus H)),$ }\\ r(H\setminus v)+ r(G\setminus H)+2 & \text{if $v[G\setminus H]_{out}\notin rs(A(G\setminus H)).$ }
\end{cases}$$

 \end{enumerate}
\end{theorem}

\begin{proof}  With suitable relabelling of the vertices, we can write
$$A(G)=\begin{bmatrix}
A(H\setminus v) & x & O\\ y^T & \alpha & w^T \\ O^T & z & A(G \setminus H)
\end{bmatrix}.$$
\begin{enumerate}
\item  $v[H\setminus v]_{in}\in cs(A(H\setminus v))$, but $v[H\setminus v]_{out}\notin rs(A(H\setminus v))$.

 In this case, the column vector $x \in cs(A(H\setminus v))$, while $y^T \notin rs(A(H\setminus v))$.
Using elementary operations $A(G)$ can be transformed to the matrix $$\begin{bmatrix}
A(H\setminus v) & 0 & O\\ y^T &  \hat{\alpha} & w^T \\ O^T & z & A(G \setminus H)
\end{bmatrix}.$$
Notice that the matrix
$$\begin{bmatrix}
A(H\setminus v) & 0 & O\\ y^T & \hat{\alpha} & w^T \end{bmatrix},$$ has rank $1+r(H\setminus v)$, and its rows are not in the row space of the matrix
$$\begin{bmatrix}
O^T & z & A(G \setminus H)
\end{bmatrix}.$$

If $z \in cs(A(G\setminus H)),$ then the rank of above matrix is $r(G\setminus H)$ otherwise $1+r(G\setminus H)$. Hence the result follows.

\item $v[H\setminus v]_{out}\in rs(A(H\setminus v))$, but $v[H\setminus v]_{in}\notin cs(A(H\setminus v))$.

In this case, the row $y^T\in rs(A(H\setminus v))$, while $x\notin cs(A(H\setminus v))$.
Using elementary operations $A(G)$ can be transformed to the matrix $$\begin{bmatrix}
A(H\setminus v) & x & O\\ 0 &  \hat{\alpha} & w^T \\ O^T & z & A(G \setminus H)
\end{bmatrix}.$$
Similar to the Case I, if $w^T\in rs(A(G\setminus H)),$ then $r(G)=r(H\setminus v)+r(G\setminus  H)+1,$ else $r(G)=r(H\setminus v)+r(G\setminus  H)+2.$
\end{enumerate} This proves the result.
\end{proof}

\begin{theorem}\label{lt}
 Let $G$ be a digraph with a cut-vertex $v$. Let $H$ be a nonempty induced  subdigraph which includes $v$ such that there is no arc $(p,q)$ or $(q,p)$, where $p\in V(H\setminus v)$ and  $ q \in V(G\setminus H)$.  If $r(H)=r(H\setminus v)+1$, and
if $v[H\setminus v]_{out} \in rs(A(H\setminus v))$ and $v[H\setminus v]_{in} \in cs(A(H\setminus v)),$ but $v[H\setminus v \ v]_{out} \notin rs([A(H\setminus v)\ v[H\setminus v]_{in}])$ and $v[H\setminus v \ v]_{in} \notin cs([A^T(H\setminus v) \ v[H\setminus v]^T_{out}]^T).$ Then

$$ r(G)=\begin{cases} \begin{aligned}
r(H\setminus v)+ r(G\setminus H)+1 & \text{ if either $v[G\setminus H]_{out}\notin rs(A(G\setminus H))$ or $v[G\setminus H]_{in} \notin cs(A(G\setminus H)),$}\\

 r(H\setminus v)+ r(G\setminus H)+2 & \text{ if $v[G\setminus H]_{out} \notin rs(A(G\setminus H))$ and $v[G\setminus H]_{in} \notin cs(A(G\setminus H)),$} \\

  r(H\setminus v)+ r(G\setminus H)+1 & \text{ if $v[G\setminus H]_{out} \in rs(A(G\setminus H))$ and $v[G\setminus H]_{in} \in cs(A(G\setminus H))$, but} \\ & \text{ $[ \hat{\alpha} \ v[G\setminus H]_{out}] \notin rs(A(G\setminus H))$ and $[\hat{\alpha} \ v[G\setminus H]_{in}]^T\notin cs(A(G\setminus H))$}.  
\end{aligned}
\end{cases}$$
\end{theorem}
\begin{proof}

In this case the matrix $A(G)$ can be transformed to the following matrix using elementary operations
$$\begin{bmatrix}
A(H\setminus v) & 0 & O\\ 0 &  \hat{\alpha} & w^T \\ O^T & z & A(G \setminus H)
\end{bmatrix},$$
from which the results are obvious.
\end{proof}

 Simple block graphs are typical examples, where the pendant blocks satisfy the condition of Theorem \ref{lt}, hence their ranks can be recursively found using Theorem \ref{lt} \cite{singh2018nonsingular}.

\bibliographystyle{plain}
\bibliography{RN}
 \end{document}